\documentclass[letterpaper, 10 pt, conference]{ieeeconf}  
\IEEEoverridecommandlockouts                             
\usepackage{cite}
\usepackage{amsmath,amssymb,amsfonts}
\usepackage{algorithmic}
\usepackage{graphicx}
\usepackage{amsmath}
\usepackage{textcomp}
\usepackage{xcolor}
\usepackage{stfloats} % Correct position of the tables
\usepackage{float}
\usepackage{accents}
\usepackage{tikz}
\usetikzlibrary{calc}
\usepackage{pgfplots}
\usepackage{hyperref}
\usetikzlibrary{shapes, arrows,automata}
\usetikzlibrary{shapes.multipart}
\usepackage[latin1]{inputenc}   
\DeclareMathAlphabet{\pazocal}{OMS}{zplm}{m}{n}
\usepackage{calrsfs}
\usepackage{dsfont}
\let\labelindent\relax
\usepackage{enumitem}
\usepackage{subfigure}
\usepackage{caption}
\usepackage{mathtools}
\usepackage{enumitem}
\usepackage{booktabs} % Consigliato per tabelle di qualitÃÂ  editoriale

\usepackage{array}
\newcolumntype{P}[1]{>{\centering\arraybackslash}p{#1}}

\newtheorem{lemma}{Lemma}

\newtheorem{definition}{Definition}

\newtheorem{proposition}{Proposition}

\newtheorem{remark}{Remark}

\newcommand{\ejo}{e^{j\omega}}
\DeclarePairedDelimiterX{\inp}[2]{\langle}{\rangle}{#1, #2}

\tikzstyle{sum} = [draw, fill=blue!20, circle, node distance=1cm]
\tikzstyle{input} = [coordinate]
\tikzstyle{pinstyle} = [pin edge={to-,thin,black}]

\usepackage{cleveref}
\usepackage{tabularx}
\usepackage{longtable} % tables that can span several pages
\usepackage{colortbl}
\usepackage[thinlines]{easytable}

\usepackage{transparent}

\usepackage{pdfpages} % To include a pdf file
\usepackage{afterpage}
\usepackage{lipsum} % DUMMY PACKAGE
\usepackage{fancyhdr} % For the headers
\usepackage{multirow} %Multi rows in tables
\usepackage{adjustbox} % Ridimensionate tables
\usepackage{amsmath} % To define subequations
\usepackage{comment} %To comment
\usepackage{color} %To highlight
\usepackage{soul} %To highlight
\usepackage{eurosym}  %To do euro symbol
\usepackage{amsfonts}

\usepackage{calrsfs}

\begin{document}
\title{\LARGE{\textbf{When Persistency is not Exciting in Data-Driven Predictive Control}}}

\author{Gianluca Giacomelli, Manuel Kl{\"a}dtke, Chuyu Lu, Siep Weiland, Moritz Schulze Darup, and Valentina Breschi
\thanks{Gianluca Giacomelli, Chuyu Lu, Siep Weiland, and Valentina Breschi are with the Control Systems Group, Eindhoven University of Technology, 5612AZ Eindhoven, The Netherlands (e-mails: \textsl{\{g.giacomelli, s.weiland, v.breschi\}@tue.nl}, \textsl{c.lu2@student.tue.nl}). Manuel Kl{\"a}dtke and Moritz Schulze Darup are with the Control and~Cyberphysical Systems Group, TU Dortmund University, 44227 Dortmund, Germany (e-mails: \textsl{\{manuel.klaedtke, moritz.schulzedarup\}@tu-dortmund.de}).}}
\maketitle

\begin{abstract}
    Understanding how to collect data that is \textquotedblleft meaningful\textquotedblright \ for control purposes is of paramount importance in data-driven control. While existing approaches have primarily relied on the satisfaction of a rank condition to assess the quality of an experiment, we show that satisfying it is not always sufficient to achieve satisfactory closed-loop performance. Focusing on scenarios where white-noise-like excitation cannot be used for data collection, we examine the frequency-domain implications of linear behavioral representation. This analysis demonstrates that leakage effects are the main driver for data to represent the dynamics of the system. These findings are reflected in our numerical results. Data-enabled predictive controllers built on datasets with insufficient bandwidth, despite fulfilling standard rank conditions, suffer from severe ill-conditioning and fail to achieve reference tracking.  
\end{abstract}


\section{Introduction}\label{sec:Intro}
Collecting good data is quintessential for the success of all learning approaches, among which lie data-driven control strategies. For linear systems, this problem has been widely studied within the system identification community, including control-oriented system identification \cite{gevers2006input,hjalmarsson2005experiment}. It is thus well known that data collection campaigns should capture the essential dynamics required to achieve one's control goal. While well-known within the system identification community, this insight has not yet been fully analyzed and exploited to guide data collection experiments for direct data-driven control. Focusing on data-driven predictive strategies \cite{coulson2019data, berberich2020data,dorfler2022bridging,yin2021maximum,breschi2023data}, data collection is indeed often practically performed using white noise as a probing input. These signals are nonetheless more hostile to the probed system's health, limiting their suitability in practice \cite{narasimhan2003multi,rivera2003plant}. Moreover, it is known that using random noise, including but not limited to white noise (see \cite[Definition 5.6]{pintelon2012system}), for data collection allows for full spectral excitation, at the price of higher crest and time factors compared to other excitation signals and leakages in the frequency measurements \cite[Section 5.3]{pintelon2012system}. Meanwhile, existing experiment-design strategies aligned with the data-enabled predictive framework \cite{coulson2019data, berberich2020data,dorfler2022bridging,yin2021maximum,breschi2023data} still focus mainly on ensuring that the data satisfy the rank condition required by Willems' Fundamental Lemma \cite{van2021beyond,alsalti2023design}. However, as showcased in \cite{coulson2022quantitative,zhou2025effective}, this rank condition alone is inherently fragile. Other techniques instead seek non-parametric, data-driven predictors that are accurate in simulation \cite{iannelli2021design}, without considering the control objective.

Considering tracking band-limited references and choosing inputs with energy concentrated in a limited frequency band, 
we investigate the frequency-domain implications of Willems' Fundamental Lemma \cite{willems2005note}. Our analysis shows that this data-driven predictor is theoretically capable of predicting inputs (and outputs) at frequencies that were explored during data collection, whereas the numerical conditioning of the corresponding predictive control problem
%prediction accuracy
is tightly linked to leakage effects for the frequencies not directly excited. Our theoretical insights are backed up by our numerical results for band-limited excitation signals, which empirically show the importance of collecting data by probing the system with inputs that are not only persistently exciting according to \cite[Section 3]{willems2005note}, but that also allow for a spectral excitation that captures what is essential to achieve one's control goal. At the same time, they highlight that guaranteeing the rank condition in \cite[Section 3]{willems2005note} is not enough to design a performing data-driven predictive controller, aligning with \cite{coulson2022quantitative,zhou2025effective} on the necessity of a well-conditioned data matrix. Expanding upon \cite{zhou2025effective}, we show that matrix conditioning cannot be assessed solely by the number of excited frequencies, but is critically governed by their location within the frequency spectrum. Through a comparison with Subspace Predictive Control (SPC) \cite{favoreel1999spc}, we further show that the performance of the latter is not influenced by the excitation band. This result is consistent with the claim that persistence of excitation (see \cite[Section 3]{willems2005note}) is sufficient when data are preprocessed, while also highlighting the need for potentially different data-collection requirements depending on whether a preliminary preprocessing step is performed.

The paper is organized as follows. Section \ref{sec:Setting} introduces the setting and goal of the paper. Section \ref{sec:freq_ana} presents the frequency-domain analysis of the data-based behavioral predictor under limited-band excitation. Our theoretical analysis is supported in Section \ref{sec:Numerical} via a numerical simulation study. Our results showcase the importance of data collections unveiling the necessary information to achieve one's control goal for a predictor without a preliminary preprocessing step. The paper closes with some concluding remarks and directions for future work in Section \ref{sec:conclusions}.

\paragraph*{Notation} We denote the set of integers, real numbers, and complex numbers as $\mathbb{Z}$, $\mathbb{R}$, and $\mathbb{C}$, respectively. The imaginary unit is $j\!\in\!\mathbb{C}$. Given a matrix $A\!\in\! \mathbb{C}^{m\times n}$, $A^{\top}$ is its transpose. If the matrix $A\in \mathbb{C}^{m\times m}$ is invertible, then its inverse is $A^{-1}$. Zero matrices are represented with $\mathbf{0}$, without explicitly indicating their dimensions for ease of notation. Meanwhile, $I_m$ indicates an identity matrix having dimensions $m\times m$. Given a signal $\zeta_t\in\mathbb{R}^{n_{\zeta}}$, with $t\in \mathbb{Z}$, then $\zeta_{[0,t]}=\begin{bmatrix} \zeta_{0}^\top & \cdots & \zeta_{t}^\top \end{bmatrix}^\top$ and the Hankel matrix filled with $\zeta_{[0,t]}$ corresponds to 
\begin{equation}\label{eq:Hankel}
    \pazocal{H}_N(\!\zeta_{\left[0, t\right]\!}) \!=\! \begin{bmatrix}
                     \zeta_{0\!} & \zeta_{1\!} & \!\cdots\! & \zeta_{t-N\!}\\
                     \zeta_{1\!} & \zeta_{2\!} & \!\cdots\! & \zeta_{t-N+1\!}\\
                     \vdots & \vdots & \!\ddots\! & \vdots\\
                     \zeta_{N-1\!} & \zeta_{N\!} & \!\cdots\! & \zeta_{t\!}
\end{bmatrix}\!\in\!\mathbb{R}^{(n_{\zeta}N)\times (t-N+1)}.
\end{equation}
Moreover, we indicate the Discrete-Time Fourier Transform (DTFT) of the time-domain signal $\zeta_t$ as $Z(e^{j\omega})$, with $\omega \!\in\! (-\pi,\pi]$. Meanwhile, the convolution product between two discrete signals $z_t$ and $h_t$ at time $t\! \in\! \mathbb{Z}$ is compactly denoted as $(z\ast h)_t$, while the convolution product between two signals $Z(\ejo)$ and $H(\ejo)$ at a frequency $\omega\! \in\! (-\pi,\pi]$ is defined as $(Z\ast H)_{\ejo}$. Lastly, $|z|$ denotes the component-wise modulus of a complex vector $z \in \mathbb{C}^{m}$, while $\|z\|_{2}$ indicates its 2-norm.

\section{Setting \& Goal}\label{sec:Setting}
Consider a controllable, Linear-Time Invariant (LTI) system described by the input/output relationship
\begin{equation}\label{eq:true_dom_t}
    y_t = G(q)u_t,~~~~ \forall t \in \mathbb{Z},
\end{equation}
where $u_t \!\in\! \mathbb{R}^m$, $y_t \!\in\! \mathbb{R}^{p}$ for all $t \in \mathbb{Z}$, while\footnote{q is the shift operator, i.e., $q^iu_{t}=u_{t+i}$ for all $i \!\in\! \mathbb{Z}$.} $G(q) \!\in\! \mathbb{R}^{p \times m}$ is assumed to be \emph{unknown}. Let this system be Bounded-Input Bounded-Output (BIBO) stable. Our goal is for the system described by \eqref{eq:true_dom_t} to track a reference signal $y_t^{\mathrm{o}} \!\in\! \mathbb{R}^{p}$ with limited support in the frequency domain, namely
\begin{equation}\label{eq:ref_specification}
    |Y^{\mathrm{o}}(\ejo)| \neq  \mathbf{0},~~~\forall\omega \in \Omega^{\mathrm{o}}\cup (- \Omega^{\mathrm{o}}),
\end{equation}
with $ \Omega^{\mathrm{o}}\!=\![\omega^{\mathrm{o}}_{\mathrm{L}},\omega^{\mathrm{o}}_{\mathrm{H}}]\!\subset\!(0, \pi)$. To this end, suppose that the system has been excited with a quasi-stationary input sequence $u^d_{[0,T_d-1]}$. Since its finite length inherently precludes band-limited excitation within $\Omega \cup (-\Omega)$, such input satisfies %\textcolor{blue}{such that}
\begin{equation}\label{eq:input_spectra}
    |U^{d}(e^{j\omega})| ~\simeq~\mathbf{0},~~~\forall \omega \notin \Omega\cup (- \Omega), 
    \end{equation}
with $\Omega$ being potentially different from $\Omega^{\mathrm{o}}$ in \eqref{eq:ref_specification} and $\simeq$ indicating the presence of residual spectral leakage preventing the spectra from being zero outside $\Omega\cup (- \Omega)$. Further assume that the corresponding outputs $y^d_{[0,T_d-1]}$ have been collected to form the dataset
\begin{equation}\label{eq:dataset_time}
    \pazocal D = \left\{u^d_{[0,T_d-1]},\, y^d_{[0,T_d-1]}\right\}.
\end{equation}
Let the order of the data-generating system described by \eqref{eq:true_dom_t} be equal to $n$ and assume that the input sequence $u^d_{[0,T_d-1]}$ is persistently exciting (PE) of order $N+n$ according to the following definition \cite{willems2005note}. 
\begin{definition}[Rank-based PE]\label{def:persistency_excitation_rank}
    The signal $u^d_{[0,T_d-1]}$ is persistently exciting of order $N+n$ if 
    \begin{equation}\label{eq:rank_condition}
        \mathrm{rank}(\pazocal{H}_{N+n}(u_{[0,T_d-1]}))=m(N+n). 
    \end{equation}
\end{definition}
\vspace{2mm}
Since the system is unknown, yet we have a set of data, our control task could be fulfilled by either identifying a model for the system and then designing an associated controller, or by using direct data-driven control strategies, skipping such explicit identification steps and using these data for control. We here employ the latter approach, focusing on data-enabled predictive control strategies. We hence aim to use the data to construct a non-parametric predictor of the form \cite{willems2005note}:
\begin{equation}\label{eq:Willems_model}
    \begin{bmatrix}
        \pazocal{H}_{N}(u_{[0,T_d-1]}^{d})\\
        \pazocal{H}_{N}(y_{[0,T_d-1]}^{d})
    \end{bmatrix}\alpha = \begin{bmatrix}
        u_{[0, N-1]}\\
        y_{[0, N-1]}
    \end{bmatrix}.
\end{equation}
Through this predictor, an input and the corresponding output trajectories of length $N$ are represented as linear combinations of the available data organized in Hankel matrices (see \eqref{eq:Hankel}), through the coefficients $\alpha\in\mathbb{R}^{T_d-N+1}$.

Based on our assumptions on the input signals used during data collection (see \eqref{eq:input_spectra}), our goals are to investigate $(i)$ whether satisfying condition \eqref{eq:rank_condition} alone is sufficient for the predictor in \eqref{eq:Willems_model} to describe the behavior of the controlled system at the control-relevant frequencies $\omega \in \Omega^{\mathrm{o}} \cup(-\Omega^{\mathrm{o}})$, and $(ii)$ how the choice of the excitation band $\Omega$ in \eqref{eq:input_spectra} impacts on the capability of a system controlled with Data-enabled Predictive Control (DeePC) \cite{coulson2019data,berberich2020data}.
\begin{remark}[Noise-free data]
    In this work, our analysis is limited to noiseless data, to better distinguish issues that can be induced by choices made during the data collection experiment from those caused by noise. The extension to the noisy case is thus left for future work.
\end{remark}

\section{Behind the Fundamental Lemma: Frequency Domain Implications}\label{sec:freq_ana}
While some data-driven predictive approaches \cite{coulson2019data, berberich2020data,dorfler2022bridging,yin2021maximum} look only at Definition \ref{def:persistency_excitation_rank} to assess the quality of data, \cite[Lemma 3]{breschi2023data} already mentions the need for spectral conditions which are more in line with what is usually employed in system identification (see, e.g., \cite{pintelon2012system,ljung1987theory,van2012subspace}). We hence start by stating the following spectral PE condition presented in \cite[Definition 13.1]{ljung1987theory}, which is also used in subspace identification (see \cite[Definition 5]{van2012subspace}).
\begin{definition}[Spectrum-based PE]\label{def:persistency_excitation_spec}
    A quasi-stationary signal $u^d_{[0,T_d-1]}$ is persistently exciting of order $N+n$ if its Power Spectral Density (PSD) matrix $\Phi^d_{u}(\ejo)\!\in\! \mathbb{C}^{m\times m}$ is positive definite for \emph{at least} $N+n$ frequencies $\omega\in(-\pi, \pi]$.
\end{definition}
\vspace{2mm}
Note that, based on \cite[Proof of Lemma 13.1]{ljung1987theory}, rank-based and spectral PE, introduced in Definitions \ref{def:persistency_excitation_rank} and \ref{def:persistency_excitation_spec}, respectively, are \emph{equivalent}. Since we assume our data to be PE of order $N+n$ according to the rank condition in \eqref{eq:rank_condition} as in our setting, this equivalence implies that the input signal used to construct $\pazocal{D}$ in \eqref{eq:dataset_time} excites at least $N+n$ frequencies in the band $\Omega \cup (-\Omega)$.

\begin{remark}[Beyond Definition \ref{def:persistency_excitation_spec}] 
    A more stringent version of Definition \ref{def:persistency_excitation_spec} is provided in \cite[Definition 13.2]{ljung1987theory}, imposing the input PSD $\Phi^d_u(e^{j\omega})$ to be positive definite for almost all frequencies. However, in our setting, out-of-band excitation relies solely on leakage, leading to an input PSD that is positive definite yet potentially negligible at unexcited frequencies.
\end{remark}

\subsection{Revisiting the non-parametric predictor as an FIR}
The properties of $u^d_{[0,T_d-1]}$, i.e., the fact that it satisfies \eqref{eq:input_spectra} and excites at least $N+n$ frequencies, allow us to interpret the input data in $\pazocal{D}$ as complex-value signals, as
\begin{equation}\label{eq:input_complex}
                u_t^{d\!}=\!\sum_{k=1}^{E}U_k^d e^{j\omega_kt}, ~\forall t \!\in\! [0,T_d\!-\!1],\mbox{ with } \omega_k \in \Omega,                        %y_t^{d\!}&=\!\!\!\sum_{k=1}^{E}\!Y_k^d e^{j\omega_kt}\!\!=\!\!\!\!\sum_{k=1}^{E}\! G(e^{j\omega_k}) U_k^de^{j\omega_kt\!} , ~\!\forall t \!\in\! [0,T_d\!-\!1],%\!=\!0,\!\ldots,\! T_d\!-\!1,
        %\label{eq:output_complex}
\end{equation}
with $E$ satisfying 
\begin{equation}\label{eq:number_excited_freq}
    \left\lceil \frac{N+n}{2}\right\rceil\leq E,
\end{equation}
and 
\begin{equation}
    U_k^d=\begin{bmatrix}
        |U_k^{d,1}|e^{j\phi_{1,k}} & \cdots &  |U_k^{d,m}|e^{j\phi_{m,k}}
    \end{bmatrix}^\top,
\end{equation}
compactly denoting the amplitude and phase of the harmonics across input channels. Note that the lower bound in \eqref{eq:number_excited_freq} is dictated by the spectral-based PE condition and the symmetry of the harmonics' spectra in $(-\pi,\pi)$.
% while the upper bound  \le \left\lfloor \frac{T_d}{2}  \right\rfloor follows from the Nyquist-Shannon sampling theorem \cite{shannon2006communication}.

Let us now define $\tilde{u}^d_{t}= u^d_{t+T_d-N}$ and $\tilde{y}^d_{t}= y^d_{t+T_d-N}$. Accordingly, by using the predictor in \eqref{eq:Willems_model}, the following holds:
\begin{subequations}
    \begin{align}
        u_t&
        =\sum_{i=0}^{T_d-N}\tilde{u}^d_{t-i} \alpha_{i}% = (\tilde{u}^d \ast f)_{t}
        ,~\forall t \in[0,N-1],\label{eq:prediction_complex_input}\\
        y_t&
        =\sum_{i=0}^{T_d-N}\tilde{y}^d_{t-i} \alpha_{i}%= (\tilde{y}^d \ast f)_{t}
        ,~\forall t \in[0,N-1].\label{eq:prediction_complex_output}
    \end{align}
Hence, the reconstructed input/output pair $(u_t,y_t)$ can be seen as responses of a Finite Impulse Response (FIR) filter, whose impulse response is given by
\begin{equation}\label{eq:impulse_response_FIR}
    f_{i}=\begin{cases}
        \alpha_{i},&\mbox{~if~}i\in[0,T_d-N],\\
        0,& \mbox{~otherwise.}
    \end{cases}
\end{equation}
\end{subequations}
This alternative perspective on \eqref{eq:Willems_model} is the key enabler to shift from the time to the frequency domain and, thus, to analyze the impact of input data exciting only a limited set of frequencies. Note that, from this point onward, we will always assume the initial state of \eqref{eq:true_dom_t} is zero.

Let us equivalently represent the dataset trajectories in   \eqref{eq:prediction_complex_input}-\eqref{eq:prediction_complex_output} as follows
%\begin{subequations}
%\begin{align}
%	& u_t=\sum_{i=0}^{+\infty} \alpha_i \textcolor{blue}{(}\tilde{w}_i \tilde{u}_{t-i}^{d,\infty}\textcolor{blue}{)},\label{eq:window_infinity_input}\\
%& y_t=\sum_{i=0}^{+\infty} \alpha_i \textcolor{blue}{(}\tilde{w}_i \tilde{y}_{t-i}^{d,\infty}\textcolor{blue}{)},\label{eq:window_infinity_output}
%\end{align}
%\end{subequations}
%as 
\begin{equation}\label{eq:windowed_signals_traj_and_data}
        \tilde{u}_{t}^{d}\!=\!\tilde{w}_t \tilde{u}_{t}^{d,\infty},~~\tilde{y}_{t}^{d}\!=\!\tilde{w}_t \tilde{y}_{t}^{d,\infty},  
\end{equation}
with $(\tilde{u}_{t}^{d,\infty},\tilde{y}_{t}^{d,\infty})$ input/output pairs from an infinite-length dataset and $\tilde{u}_{t}^{d,\infty}$ such that
\begin{equation}
        |\tilde{U}^{d,\infty}(e^{j\omega})| =\mathbf{0},~~~\forall \omega \notin \Omega\cup (- \Omega), 
\end{equation}
%\textcolor{blue}{which is the counterpart of \eqref{eq:input_spectra} achieved with an infinite-length signal,} 
while 
%\begin{equation}\label{eq:window_time}
%	w_{t}=\begin{cases}
		%1, &\mbox{~if~} t \in [0,N-1],\\
		%0, & \mbox{~otherwise,}
	%\end{cases}
%\end{equation}
%and
\begin{equation}
	\tilde{w}_{t}=\begin{cases}
		1, &\mbox{~if~} t \in [-(T_d-N),N-1],\\
		0, & \mbox{~otherwise.}
	\end{cases}
\end{equation}
%Under the assumption that the Finite Impulse Response (FIR) filter characterizing \eqref{eq:window_infinity_input}-\eqref{eq:window_infinity_output} is BIBO stable, then we can express \eqref{eq:window_infinity_input}-\eqref{eq:window_infinity_output} in the frequency domain as follows
%\begin{subequations}
%	\begin{align}
%		& U(e^{j\omega})=F^{W}(e^{j\omega})\tilde{U}^{d,\infty}(e^{j\omega}),\label{eq:prediction_freq_input_window}\\
%		& Y(e^{j\omega})=F^{W}(e^{j\omega}) \tilde{Y}^{d,\infty}(e^{j\omega}),\label{eq:prediction_freq_output_window}
%	\end{align}
%\end{subequations}
%with
%\begin{equation}\label{eq:windowing_dataset}
%F^{W}(e^{j\omega})=\sum_{i=0}^{+\infty}\alpha_{i}\tilde{w}_{i}e^{-j\omega i},
%\end{equation}
% and
By looking at \eqref{eq:prediction_complex_input}-\eqref{eq:prediction_complex_output} through their FIR representation prior to their restriction to the finite prediction horizon, the predictor in \eqref{eq:Willems_model} can be expressed in the following frequency-domain form
\begin{equation}\label{eq:FRF_conv1}
    \begin{aligned}
    U(e^{j\omega})&\!=\! F(\ejo)\tilde{U}^{d}(\ejo),%=F(\ejo) (\tilde{U}^{d,\infty}\ast \pazocal{W}){_{\ejo}},
    ~~Y(e^{j\omega})\!=\! F(\ejo)\tilde{Y}^{d}(\ejo),%=F(\ejo) (\tilde{Y}^{d,\infty}\ast \pazocal{W}){_{\ejo}},
    \end{aligned}
\end{equation}
where $\tilde{U}^{d}(\ejo)$ and $\tilde{Y}^{d}(\ejo)$ are the DTFTs of the input/output trajectories in $\pazocal{D}$, while
\begin{equation}\label{eq:FRF_FIR}
    F(\ejo)=\!\!\! \sum_{i=0}^{T_d-N} \alpha_i e^{-j\omega i}.
\end{equation}
Note that the DTFTs of $\tilde{u}_t^d$ and $\tilde{y}_t^d$ have the same component-wise moduli as those of $u_t^d$ and $y_t^d$, as the former are time-shifted versions of the latter. Using this property and applying the triangle inequality to \eqref{eq:FRF_conv1}, it then follows that
    \begin{subequations}\label{eq:willems_bounds}
    \begin{align}\label{eq:willems_bounds_input} \big|U(\ejo)\big| &\!=\!|F(\ejo)|\big| {U}^d(e^{j\omega}) \big| \!\leq
\! \|\alpha\|_1\big| {U}^d(e^{j\omega}) \big|,\\
        \big|Y(\ejo)\big| &\!=\!|F(\ejo)|\big| {Y}^d(e^{j\omega}) \big|\!\leq
\! \|\alpha\|_1 \big| {Y}^d(e^{j\omega}) \big|,\label{eq:willems_bounds_output}
    \end{align}
    \end{subequations}
since  
\begin{equation}\label{eq:magnitudebound}
|F(\ejo)|\leq
\sum_{i=0}^{T_d-N}|\alpha_i|= \|\alpha\|_1.    
\end{equation}

Meanwhile, by leveraging the signal representation in \eqref{eq:windowed_signals_traj_and_data}, it further holds that
\begin{equation}\label{eq:FRF_conv}
    \begin{aligned}
    \tilde{U}^d(e^{j\omega})&=\frac{1}{2\pi}(\tilde{U}^{d,\infty}\ast \pazocal{W}){_{\ejo}},\\
    \tilde{Y}^d(e^{j\omega})&=\frac{1}{2\pi} (\tilde{Y}^{d,\infty}\ast \pazocal{W}){_{\ejo}},
    \end{aligned}
\end{equation}
with
%\textcolor{blue}{By taking the DTFT of  \eqref{eq:prediction_complex_input}-\eqref{eq:prediction_complex_output} and using the signal representation in \eqref{eq:windowed_signals_traj_and_data}, the predictor in \eqref{eq:Willems_model} can be expressed in the following frequency-domain form}
%\begin{equation}\label{eq:FRF_conv}
%    \begin{aligned}
%    U(e^{j\omega})&= %\textcolor{blue}{F(\ejo)\tilde{U}^{d}(\ejo)=F(\ejo) (\tilde{U}^{d,\infty}\ast \pazocal{W}){_{\ejo}},}\\
%    Y(e^{j\omega})&= \textcolor{blue}{F(\ejo)\tilde{Y}^{d}(\ejo)=F(\ejo) (\tilde{Y}^{d,\infty}\ast \pazocal{W}){_{\ejo}},}
%    \end{aligned}
%\end{equation}
%where %\textcolor{blue}{ $F(\ejo)$ is the FRF of an FIR given by
%\begin{equation}\label{eq:FRF_FIR}
%    F(\ejo)= \sum_{i=0}^{T_d-N} \alpha_i e^{-j\omega i},
%\end{equation}%}
%\textcolor{blue}{while $\tilde{U}^{d}(\ejo)$ and $\tilde{Y}^{d}(\ejo)$} are the FRFs of the windowed input/output \textcolor{blue}{dataset} signals, with
\begin{equation}\label{eq:window_freq}
     \pazocal{W}(\ejo)=
e^{j\omega\frac{T_d-2N+1}{2}}
\frac{\sin\left(\frac{\omega T_d}{2}\right)}
{\sin\left(\frac{\omega}{2}\right)}. %e^{-j\omega \frac{\textcolor{blue}{T_d-N}}{2}}\frac{\mathrm{sin}\left(\frac{\omega \textcolor{blue}{(T_d-N+1)}}{2}\right)}{\mathrm{sin}\left(\frac{\omega}{2}\right)}.
\end{equation}
Due to the windowing-induced leakage and by leveraging once more the fact that $\|\tilde{U}^d(e^{j\omega})\|_{2}\!=\!\|{U}^d(e^{j\omega})\|_{2}$ and $\|\tilde{Y}^d(e^{j\omega})\|_{2}=\|{Y}^d(e^{j\omega})\|_{2}$, the following inequalities hold
\begin{subequations}\label{eq:leakagebounds}
        \begin{align}
        & \| {U}^d(e^{j\omega}) \|_2 \!\le 
\sup_{\theta \in \boldsymbol{\Omega}} \big|\pazocal{W}(e^{j(\omega-\theta)})\big| \frac{1}{2\pi}\!\!\int_{\boldsymbol{\Omega}}\!\! \|\tilde{U}^{d,\infty}(e^{j\theta})\|_2 d\theta,
\label{eq:leakagebound1}\\
        & \| {Y}^d(e^{j\omega}) \|_2 \!\le\!
\sup_{\theta \in \boldsymbol{\Omega}} |\pazocal{W}(e^{j(\omega-\theta)})| \frac{1}{2\pi}\!\!\int_{\boldsymbol{\Omega}}\!\! \|\tilde{Y}^{d,\infty}(e^{j\theta})\|_2d\theta,
\label{eq:leakagebound2}
    \end{align}    
\end{subequations}
where we have used the triangle inequality and, for the sake of compactness, we denoted $\boldsymbol{\Omega}=\Omega\cup(-\Omega)$.
%which leads to the following result by using the triangle inequality in \eqref{eq:FRF_conv}.}
%\textcolor{blue}{\begin{lemma}\label{lemma:leak}
%    Given a dataset $\pazocal{D}$ with trajectories that are PE according to Definition \ref{def:persistency_excitation_rank}, the FRFs of the signals reconstructed by the predictor in \eqref{eq:Willems_model} satisfy the following
%    \begin{subequations}
%    \begin{align}\label{eq:willems_bounds_input} \big|U(\ejo)\big| &\leq \big|F(\ejo)\big| \big| U^d(e^{j\omega}) \big|,\\
%        \big|Y(\ejo)\big| &\leq \big|F(\ejo)\big| \big| Y^d(e^{j\omega}) \big|,\label{eq:willems_bounds_output}
%    \end{align}
%    \end{subequations}
%    where $\big| U^d(e^{j\omega}) \big|$ and $\big|Y(\ejo)\big|$ are upper bounded as in \eqref{eq:leakagebound1} and \eqref{eq:leakagebound2}, respectively.
%\end{lemma}}
%\begin{proof}
%    \textcolor{blue}{The proof follows from \eqref{eq:prediction_complex_input}-\eqref{eq:prediction_complex_output} to \eqref{eq:leakagebound1}-\eqref{eq:leakagebound2}, using the triangle inequality in \eqref{eq:FRF_conv}.}
%\end{proof}
Combining \eqref{eq:willems_bounds} and \eqref{eq:leakagebounds} reveals a practical limitation of the non-parametric predictor in \eqref{eq:Willems_model} under limited excitation. In particular, \eqref{eq:leakagebounds} shows that the available spectral content of the dataset for frequencies $\omega \notin \Omega\cup(-\Omega)$ is determined by the leakage induced by the finite observation window and is governed by the side-lobe decay of the window spectrum. Consequently, reconstructing a trajectory with non-negligible spectral content at such frequencies requires an increasingly large magnitude of $F(e^{j\omega})$ to compensate for the reduced spectral content of the data at such frequencies. In the limiting case in which the residual leakage vanishes, a non-zero spectral component outside the excited frequency set would require $|F(e^{j\omega})|$ to become unbounded. In turn, due to \eqref{eq:magnitudebound}, this entails an unbounded growth of the 1-norm of the predictor coefficients. Hence, although finite-window leakage introduces non-zero spectral content outside the set of excited frequencies, the computation of the predictor coefficients may become increasingly ill-conditioned as the residual leakage decreases.

\subsection{Links with frequency-domain state-space representations}\label{sec:evaluation_of_FRF}
Starting from a state-space representation of the system inspired by that adopted in \cite[Section 2.1]{vakilzadeh2015experiment}, we now aim to show if our previous considerations still hold. To this end, let us consider the following state-space representation of the system in the frequency domain:
\begin{subequations}\label{eq:freq_equations}
    \begin{align}
        e^{j\omega} X(\ejo) &= A X(\ejo)+B U(\ejo),\label{eq:state_eq_freq}\\
        Y (\ejo) & = C  X(\ejo) + D U(\ejo),\label{eq:output_eq_freq}
    \end{align}
where $X(\ejo)\in\mathbb{C}^n$ is a minimal state of the system in the frequency domain, and the matrices $\{A,B,C,D\}$ satisfy 
\begin{equation}
    G(\ejo) = C(\ejo I_n-A)^{-1}B+D,~~~\forall \omega \in (-\pi,\pi].
\end{equation}
\end{subequations}
Note that $G(e^{j\omega})$ is linked with the representation in \eqref{eq:true_dom_t}, which is obtained by considering $q=e^{j\omega}$. By %assuming that the initial state of the system is zero and that the system is observable, \eqref{eq:state_eq_freq}-\eqref{eq:output_eq_freq} can be used 
applying \eqref{eq:freq_equations} recursively, one gets %leading to 
\begin{subequations}\label{eq:frequency_eq}
\begin{equation}\label{eq:data_equation_state}
    W_L(\ejo)\! \otimes\! Y(\ejo) \!=\! \pazocal{O}_{L} X(\ejo) \! +\!\pazocal{T}_L(W_L(\ejo)\!\otimes\! U(\ejo)),
\end{equation}
where $\otimes$ denotes the Kronecker product, $\pazocal{O}_{L}$ and $\pazocal{T}_L$ are the observability and Toeplitz matrices associated with \eqref{eq:freq_equations}, i.e.,
\begin{equation*}
    \pazocal{T}_{L}\!=\!\! 
    \begin{bmatrix}
        D & \mathbf{0} & \dots & \dots & \mathbf{0} \\
        CB & D & \dots & \dots & \mathbf{0} \\
        CAB & CB & D & \dots & \mathbf{0}\\
        \vdots & \vdots & \ddots & \ddots & \vdots \\
        CA^{L-2}B\!\! & \dots & \dots & CB & D
    \end{bmatrix}\!\!,~
    \pazocal{O}_{L} \!=\!\begin{bmatrix}
        C\\
        CA\\
        CA^2\\
        \vdots\\
        CA^{L-1}
    \end{bmatrix}\!\!,
\end{equation*}
while
\begin{equation}
    W_L(\ejo) = \begin{bmatrix}
        1&\cdots & e^{j\omega(L-1)}
    \end{bmatrix}^\top.
\end{equation}
Assuming that the initial state of the system is zero, $X(\ejo)$ is defined as
\begin{equation}\label{eq:FRF_data_state}
    X(\ejo) = (\ejo I_n-A)^{-1}B U(\ejo),
\end{equation}
\end{subequations}
which further implies that $A$ has no eigenvalues on the unit circle. Further assuming that the system is observable, we can then exploit the fact that there always exists a matrix $\pazocal{C}_{det}$ such that
\begin{equation}\label{eq:state_past_data_time}
    x_t = \underbrace{\begin{bmatrix}
        \pazocal{C}_u & \pazocal{C}_y
    \end{bmatrix}}_{=\pazocal{C}_{\mathrm{det}}} \begin{bmatrix}
        u_{[t-\rho,t-1]}\\
        y_{[t-\rho,t-1]}
    \end{bmatrix},
\end{equation}
for $\rho\geq n$ (see \cite[Remark 1]{breschi2023data}). Accordingly, we can further recast $X(e^{j\omega})$ as
\begin{equation}\label{eq:state_past_data_freq}
    X(\ejo) \!=\! \pazocal{C}_u (W_{\rho}(\ejo)\!\otimes\! U(\ejo)) \!+\! \pazocal{C}_y (W_{\rho}(\ejo)\!\otimes\! Y(\ejo)),
\end{equation}
where, with a slight abuse of notation, we introduce
\begin{equation}\label{eq:past_window_ejo}
    W_{\rho}(\ejo) = \begin{bmatrix}
        e^{-j\omega\rho}&\cdots & e^{-j\omega}
    \end{bmatrix}^\top.
\end{equation}
Combined, these results lead to
\begin{align}\label{eq:no_window_eq_extended}
        W_L \!&\otimes\! Y =\Gamma_L \begin{bmatrix}
W_\rho \otimes U\\
W_\rho \otimes {Y}
\end{bmatrix}+ \pazocal{T}_L(W_L\!\otimes\! {U}), 
\end{align}
where we have dropped the explicit dependence on $e^{j\omega}$ for the sake of compactness and $\Gamma_L=\pazocal{O}_{L} \pazocal{C}_{\mathrm{det}}$. By replacing the equalities in \eqref{eq:FRF_conv1}, we ultimately get
\begin{align}\label{eq:F_extended}
        W_L \!&\otimes\! (F \tilde{Y}^d) \!=\!\Gamma_L \!\!\begin{bmatrix}
W_\rho \!\otimes\! (F \tilde{U}^d)\\
W_\rho \!\otimes\! (F \tilde{Y}^d)
\end{bmatrix}\!\!+\! \pazocal{T}_L[W_L\!\otimes\!( F \tilde{U}^d)], 
\end{align}
which can be equivalently recast as
\begin{align}\label{eq:F_extended2}   
        \nonumber W_L \!\otimes (W_{K}^{-}\alpha\tilde{Y}^d) =&\Gamma_L\! \begin{bmatrix}
W_\rho \!\otimes\! (W_{K}^{-}\alpha \tilde{U}^d)\\
W_\rho \!\otimes\! (W_{K}^{-}\alpha \tilde{Y}^d)
\end{bmatrix}\\
&+ \pazocal{T}_L[W_L\!\otimes\! (W_{K}^{-}\alpha \tilde{U}^d)], 
\end{align}
where 
\begin{equation}
    W_{K}^{-}(\ejo)\!=\!\begin{bmatrix}
        1 & e^{-j\omega} &\!\! \cdots\!\! & e^{-j\omega(K-1)}
    \end{bmatrix},
\end{equation}
and $K=T_d-N+1$.

Let us further introduce
\begin{subequations}\label{eq:DeePC_quantities}
\begin{align}
    & U_{P}= W_\rho\!\otimes\! (W_{K}^{-}\alpha \tilde{U}^d),~~U_{F}=W_L\!\otimes\! (W_{K}^{-}\alpha \tilde{U}^d),\label{eq:DeePC_quantities_U}\\
     & Y_{P}=W_\rho \!\otimes\! (W_{K}^{-}\alpha \tilde{Y}^d),~~Y_{F}=W_L\!\otimes\! (W_{K}^{-}\alpha \tilde{Y}^d),
\end{align}
\end{subequations}
which are the frequency-domain equivalents of the quantities conventionally used in the DeePC predictor~\cite{coulson2019data, berberich2020data,dorfler2022bridging,yin2021maximum,breschi2023data}. In turn, these latter allow us to formalize the following result.
\begin{proposition}\label{th:consistency}
    For any $\alpha\in\mathbb{R}^{T_d-N+1}$, the following bounds hold:
\begin{subequations}\label{eq:bound_proprosition}
\begin{align}
    & \|\Xi_P(\ejo)\|_{2}\leq \sqrt{\rho}\|\alpha\|_1\|\Xi^d(\ejo)\|_{2},\label{eq:bound_proprosition_past}\\
    & \|\Xi_F(\ejo)\|_{2}\leq \sqrt{L}\|\alpha\|_1\|\Xi^d(\ejo)\|_{2},
\end{align}
\end{subequations}
where $\|\alpha\|_{1}$ denotes the $1$-norm of $\alpha$ in \eqref{eq:Willems_model}, $\Xi$ is a placeholder for either $U$ or $Y$, while $U_{P}(\ejo)$, $Y_{P}(\ejo)$, $U_{F}(\ejo)$ and  $Y_{F}(\ejo)$ are defined in \eqref{eq:DeePC_quantities}.   
\end{proposition}
\begin{proof}
    For simplicity, we prove the bound for $U_P$ in \eqref{eq:DeePC_quantities_U}. The other bounds follow straightforwardly. By exploiting the properties of the Kronecker product, we have that 
        \begin{equation*}
            \|U_{P}\|_2 = \|W_\rho\|_2 \| W_{K}^{-}\alpha \tilde{U}^d\|_2 \leq \|W_\rho\|_2 \| W_{K}^{-}\alpha\|_2 \|\tilde{U}^d\|_2,
        \end{equation*}
        with the upper bound following from the submultiplicativity of the 2-norm. Leveraging once again the equality $\|{U}^d(\ejo)\|_{2}=\|\tilde{U}^d(\ejo)\|_{2}$ and expanding on this last expression, we obtain the bound in \eqref{eq:bound_proprosition_past} for $U_P(\ejo)$ since
        \begin{equation*}
            \|U_{P}\|_2 \leq \|W_\rho\|_2 \| W_{K}^{-}\alpha\|_2 \|U^d\|_2 \leq \sqrt{\rho} \|\alpha\|_1 \|U^d\|_2,        
        \end{equation*}
        where we used the definition of $W_\rho$ in  \eqref{eq:past_window_ejo} to compute its 2-norm and \eqref{eq:magnitudebound} for upper bounding $\| W_{K}^{-}\alpha\|_2= | W_{K}^{-}\alpha |$.
\end{proof}
Combining this result with the bounds in \eqref{eq:leakagebounds}, we once again observe that predicting trajectories with non-negligible spectral content at unexcited frequencies becomes increasingly ill-conditioned as the effect of leakages decreases. Indeed, when the 2-norms of the available data at $\omega \notin \boldsymbol{\Omega}$ become negligible and yet we are interested in reconstructing a non-zero signal at such a frequency, then $\|\alpha\|_{1}$ required to do so needs to approach infinity.

\section{Illustrative Examples}\label{sec:Numerical}
\begin{table}[!tb]
\centering
\caption{Parameters of the LTI systems in \eqref{eq:sys_family}.}
\label{tab:sys_params}
    \begin{tabular}{l c c c c c}
    \hline
     & \! System 1 & \! System 2 & System 3 & System 4 & System 5\\
    \hline
    $\theta_1$ & 0 & 0 & 0 & 1.000 & 1.000\\
    $\theta_2$ & 0.009 & 0.200 & 0  & -1.176 & -0.219\\
    $\theta_3$ & 0.009 & -0.370 & -0.400  & 1.000 & 0.490\\
    $\theta_4$ & 0 & 0 & 0  & 1.000 & 1.000\\
    $\theta_5$ & 0 & 0 & 0  & -2.203 & -1.012\\
    $\theta_6$ & 1.000 & 1.000 & 1.000  & 2.296 & 0.312\\
    $\theta_7$ & -1.724 & -1.600 & -1.600  & -1.496 & -0.338\\
    $\theta_8$ & 0.741 & 0.698 & 0.870 & 0.518 & 0.508\\
    \hline
    \vspace{0.2mm}
    \end{tabular}
\end{table}

\begin{table}[!tb]
    \centering
    \caption{Upper and lower limits of the excitation bandwidth $\Omega=[\omega_{\mathrm{L}},\omega_{\mathrm{H}}]$ excited with a multisine input $u^d_{[0,T_d-1]}$.}
    \label{tab:scenario_freq_range}
    \begin{tabular}{c c c c c c c c c}
    \hline 
    & WN & IB & IBN & IBW & OB & OBN & OBW \\
    \hline 
    $\omega_{\mathrm{L}}$ & 0 & $0.01\pi$ & $0.01\pi$ & $0.01\pi$ & $0.3\pi$ & $0.4\pi$ & $0.35\pi$ \\
    $\omega_{\mathrm{H}}$ & $\pi$ & $0.51\pi$ & $0.37\pi$ & $0.61\pi$ & $0.8\pi$ & $0.76\pi$ & $0.95\pi$ \\
    \hline 
     \end{tabular}
\end{table}
%
%
%
%
%\begin{figure*}[!tb]
%\centering
%\begin{tabular}{cccc}
%    \includegraphics[scale=0.44]{Img/u_S5_Mean_CI95_WN.pdf} &   \includegraphics[scale=0.44]{Img/u_S5_Mean_CI95_IBW.pdf}&
%    \includegraphics[scale=0.44]{Img/u_S5_Mean_CI95_IBN.pdf} & \includegraphics[scale=0.44]{Img/u_S5_Mean_CI95_OB.pdf}\\
%    \includegraphics[scale=0.44]{Img/y_S5_Mean_CI95_WN.pdf} &   \includegraphics[scale=0.44]{Img/y_S5_Mean_CI95_IBW.pdf}&
%    \includegraphics[scale=0.44]{Img/y_S5_Mean_CI95_IBN.pdf} & \includegraphics[scale=0.44]{Img/y_S5_Mean_CI95_OB.pdf}
%\end{tabular}
 %   \caption{(\emph{DeePC}) Tracking performance: mean values (colored line) and standard deviation (shaded area) achieved in closed-loop over $50$ simulations with different initial conditions. Outputs are compared against the reference to be tracked (black).}
  %  \label{fig:mean_CI_traj_DeePC}
%\end{figure*}
%
%\begin{figure*}[!tb]
%\centering
%\begin{tabular}{ccc}
%    \includegraphics[scale=0.44]{Img/u_S5_Mean_CI95_WN_zoom.pdf} &   \includegraphics[scale=0.44]{Img/u_S5_Mean_CI95_IBW_zoom.pdf}&
%    \includegraphics[scale=0.44]{Img/u_S5_Mean_CI95_IBN_zoom.pdf}\\
%    \includegraphics[scale=0.44]{Img/y_S5_Mean_CI95_WN_zoom.pdf} &   \includegraphics[scale=0.44]{Img/y_S5_Mean_CI95_IBW_zoom.pdf}&
%    \includegraphics[scale=0.44]{Img/y_S5_Mean_CI95_IBN_zoom.pdf}
%\end{tabular}
    %\caption{(\emph{DeePC}) Zoomed-in view of the first $20$ steps of the closed-loop simulations for the WN, IBW, and IBN cases.}
    %\label{fig:mean_CI_traj_DeePC_zoom}
%\end{figure*}
%
Let us now consider\footnote{The code to reproduce our experiments is available at \url{https://github.com/GiacomelliGianluca/Not_Exciting}.} a set of LTI systems %whose input/output relationship is represented as in \eqref{eq:true_dom_t}, with these systems 
characterized by %as
\begin{equation}\label{eq:sys_family}
    G(q)=\frac{\theta_1q^2+\theta_2q+\theta_3}{\theta _4q^4+ \theta_5q^3+\theta_6q^2+\theta_7q+\theta_8},
\end{equation}
for the parameters reported in Table \ref{tab:sys_params}.

Our control goal for each system is to track an output reference $y_t^{\mathrm{o}}=\sin(\omega^{\mathrm{o}}t)$ with $\omega^{\mathrm{o}}=0.04\pi$, of which we assume to have a preview. To this end, we collect $T_d=200$ steady-state input/output samples by exciting each system with a Schroeder multisine input (see \cite[Section 5.3.1.2]{pintelon2012system}) to excite frequencies within a limited band $\Omega=[\omega_L,\omega_H]$. As shown in Table~\ref{tab:scenario_freq_range}, these frequencies are selected to compare with a baseline white noise (WN) excitation\footnote{We extract the probing signal samples from a standard normal distribution.} six scenarios, of which three are characterized by $\omega^\mathrm{o} \in \Omega$ (i.e., $\omega^\mathrm{o}$ is In-Band) and the other three feature $\omega^\mathrm{o} \notin \Omega$ (namely, $\omega^\mathrm{o}$ is Out-of-Band). Among them, we consider three cases: a nominal one, where the two scenarios for the tracked reference are indicated as In Band (IB) and Out-of-Band (OB); a scenario where the excited band is narrower than the nominal one, denoted as IB Narrow (IBN) and OB Narrow (OBN); and a case in which the excited band is wider than the nominal one, namely IB Wide (IBW) and OB Wide (OBW).%and the excitation band is relatively large (\textcolor{blue}{In Band (IB) and Out of Band (OB), respectively});  $\omega^\mathrm{o} \in \Omega$ \textcolor{blue}{or $\omega^\mathrm{o} \notin \Omega$} and we narrow the excitation band (IB Narrow (IBN) \textcolor{blue}{and OB Narrow (OBN), respectively}); \textcolor{blue}{$\omega^\mathrm{o} \in \Omega$ or $\omega^\mathrm{o} \notin \Omega$ with a wider excitation band (IB Wide (IBW) and OB Wide (OBW), respectively);} 
%These are compared against a white noise (WN) excitation\footnote{We extract the probing signal samples from a standard normal distribution.}.
\\
The available data are used to construct the data-driven predictor in \eqref{eq:Willems_model} for $N=L+\rho$, where $L=15$ %represents the prediction horizon, 
and $\rho=8$. The resulting predictor is in turn embedded into a DeePC scheme (see~\cite[Section V.C]{coulson2019data}), constraining the inputs to remain in $\mathbb{U}\!=\![-2,2]$ and minimizing the squared tracking error and the input over the prediction horizon. These two objectives are respectively weighted by $\pazocal{Q}=100I_p$ and $\pazocal{R}=I_m$. %is the past horizon used to characterize the initial conditions of the DeePC problem. In particular, by recasting the input and output Hankel data matrices in the predictor as in \eqref{eq:Hankels_for_DeePC}, we solve at each instant $t$ the following optimization problem
%\begin{subequations}\label{eq:DeePC}
%\begin{align}
%&\underset{\substack{u_{f},~y_{f},~\alpha_t}}{\mathrm{minimize}}~~J(u_{f},y_{f};y_f^{\mathrm{o}}(\omega^{\mathrm{o}}))\\
%&\qquad~~\mbox{s.t. }~\begin{bmatrix}
%    U_{P}\\
%    Y_{P}\\
%    U_{F}\\
%    Y_{F}
%\end{bmatrix}\alpha_t=\begin{bmatrix}
%z_{\mathrm{ini},t}\\
%u_{f}\\
%y_{f}
%\end{bmatrix},\label{eq:DeePC_predictor}\\
%& \qquad \qquad ~~~ u_{k} \in \mathbb{U},~~\forall k \in [0,L-1],
%\end{align}
%	where $\mathbb{U}\!=\![-2,2]$, $u_{f}\!=\!u_{[0,L-1]}$ and $y_{f}\!=\!y_{[0,L-1]}$ are the predicted inputs/outputs, $y_{f}^{\mathrm{o}}(\omega^{\mathrm{o}})\!=\!y_{[0,L-1]}^{\mathrm{o}}(\omega^{\mathrm{o}})$ is the previewed reference signal over the prediction horizon,   
%\begin{equation}
   % z_{\mathrm{ini},t}=\begin{bmatrix}
   %     u_{[t-\rho,t-1]}^{\top} & y_{[t-\rho,t-1]}^{\top} 
   % \end{bmatrix}^{\top}\!\in \mathbb{R}^{(m+p)\rho}
%\end{equation}
%is the trajectory of past inputs and outputs of length $\rho$ representing the initial condition, and the cost function is
%\begin{equation}\label{eq:DeePCloss}
%		J(u_f,y_f;y_f^{\mathrm{o}}(\omega^{\mathrm{o}}))\!=\!\!\sum_{k=0}^{L-1}\!\left[\|y_{k}\!-\!y_{k}^{\mathrm{o}}(\omega^{\mathrm{o}}) \|_{\pazocal{Q}}^{2\!}\!+\!\|u_{k}\|_{\pazocal{R}}^{2}\right]\!,
%\end{equation}
%with $\pazocal{Q}=100I_p$ and $\pazocal{R}=I_m$. 
%\end{subequations}
%Once the problem is solved, we apply the first $m$ elements of $u_f$ to the actual system, update the initial condition based on its response, and then solve the problem again in a receding horizon fashion. 
The DeePC problem is solved using Gurobi~\cite{gurobi} in the CVX optimization environment~\cite{grant2008cvx}.

In line with \cite{coulson2022quantitative, zhou2025effective}, we evaluate how the numerical conditioning of the DeePC problem changes with different excitations via 
\begin{equation}\label{eq:cond_number}
    \kappa = \frac{s_1}{s_{mN+n}},
\end{equation}
where $s_1$ is the maximum singular value of the Hankel data matrix in \eqref{eq:Willems_model} and $s_{mN+n}$ corresponds to the $m(N+n)$-th singular value at the PE threshold (see Definition~\ref{def:persistency_excitation_rank}). Note that, unlike the metrics introduced in \cite{coulson2022quantitative, zhou2025effective}, which are used to isolate the dynamic content in the data from measurement noise, this index restricts the evaluation to the subspace required for the DeePC problem to be well-conditioned. This choice is driven by two facts. First, the prediction accuracy of \eqref{eq:Willems_model} depends on its spanned trajectory subspace, which is governed by the magnitude of its first $m(N+n)$ singular values (see \cite{coulson2022quantitative}). Second, less spectrally rich excitation signals induce a premature decay of the dominant singular values before reaching the $m(N+n)$-th index, as pointed out in \cite[Section III]{zhou2025effective}.
%
%\textcolor{blue}{\textcolor{brown}{we evaluate whether the collected data are sufficiently well-conditioned for the DeePC problem to be solved  the condition number across the tested setups as follows}
% Unlike full-dimension metrics, which are used to isolate the dynamic content from measurement noise as in \cite{coulson2022quantitative, zhou2025effective}, this index restricts the evaluation strictly \textcolor{brown}{to the subspace required for the DeePC problem to be well-conditioned.}}

Moreover, we compare the performance achieved with the different data collection schemes over a simulation horizon of $T_c=100$ steps, considering $50$ initial conditions obtained by applying input sequences of length $\rho$ randomly drawn from a uniform distribution in the interval $[-1,1]$ and memorizing the corresponding outputs. This comparison is qualitatively performed by computing the mean $\mu_{\text{RMSE}}$ and the standard deviation (std) $\sigma_{\text{RMSE}}$ of the Root Mean Square tracking Error (RMSE) for the feasible solutions, i.e.,
\begin{equation}\label{eq:RMSE}
    \mathrm{RMSE}_{y}=\sqrt{\frac{1}{T_c}\sum_{t=0}^{T_c-1}(y_t-y_{t}^{\mathrm{o}})^{2}},
\end{equation} 
for each of the initial conditions tested. Furthermore, we monitor the percentage of infeasible closed-loop optimization problems $\pazocal{I}$, where encountering an infeasibility prematurely terminates the simulation. % \textcolor{blue}{and computing its mean and standard deviation $\mu_{\text{RMSE}}$ and $\sigma_{\text{RMSE}}$, respectively.} 
As shown in Table~\ref{tab:DeePC_res}, satisfying the PE condition alone does not guarantee uniform performance, with tracking errors and infeasibility counts that vary significantly depending on the excitation bandwidth. Indeed, matrix conditioning depends not only on the number of excited frequencies \cite{zhou2025effective}, but also on their spectral location. In particular, signals that excite the frequencies corresponding to the singularities of \eqref{eq:sys_family} yield superior conditioning. In turn, this results in performance as good as the WN ones. Note that this is possible even if the excited bandwidth does not enclose the system's singularity, provided it is \textquotedblleft close enough\textquotedblright \ to them and not in low-gain frequency regions (which is instead what happens for, e.g., OBW for System 1). Further analyzing those excitation bands that lead to ill-conditioned DeePC schemes reveals a trend in their feasibility behavior. A comparison between IB against OBW for System 1, or IBN against OBN across the other setups, shows that spectral proximity to the target frequency does not guarantee feasibility across all simulations, but it significantly reduces the occurrence of infeasibility in them (see the last column of Table \ref{tab:DeePC_res}). In most of the studied cases, target proximity thus seems to play a role in recovering solver feasibility. At the same time, the results attained with System 5 reveal that proximity alone is insufficient if the excitation bandwidth is too narrow (see the results obtained with IBN and IB for this system). Hence, although ill-conditioning prevents the input from practically spanning the entire trajectory subspace, near-target excitation with \textquotedblleft adequate\textquotedblright \ spectral width allows the data matrix to retain enough essential information to contain infeasibility.
\begin{table}[!tb]
\centering
\caption{(\emph{DeePC}) Condition number~\eqref{eq:cond_number}, mean, and std of RMSEs~\eqref{eq:RMSE}, and infeasibility percentage across systems and initial conditions for the different excitations.}
\label{tab:DeePC_res}
{
\begin{tabular}{l | c | c | c c c}
\hline
Case & Sys. & $\kappa$ & $\mu_{\text{RMSE}}$ & $\sigma_{\text{RMSE}}$ & $\pazocal{I}$ \\
\hline
WN      & \multirow{7}{*}{1} & \textbf{\textcolor{green!30!black}{14.701}} & 0.077 & \textbf{\textcolor{green!30!black}{0.021}} & \textbf{\textcolor{green!30!black}{0}}  \\
IB      &                    & $7.320 \times 10^{7}$  & \textbf{\textcolor{green!30!black}{0.076}} & 0.022 & 2\%  \\
OB      &                    & $5.810 \times 10^{6}$  & 0.077 & \textbf{\textcolor{green!30!black}{0.021}} & \textbf{\textcolor{green!30!black}{0}}  \\
IBN     &                    & \textbf{\textcolor{red!40!black}{$2.236 \times 10^{11}$}} & \textbf{\textcolor{red!40!black}{0.098}} & \textbf{\textcolor{red!40!black}{0.029}} & 14\%  \\
OBN     &                    & $2.102 \times 10^{9}$  & --    & --    & \textbf{\textcolor{red!40!black}{100\%}} \\
IBW     &                    & $7.115 \times 10^{5}$  & 0.077 & \textbf{\textcolor{green!30!black}{0.021}} & \textbf{\textcolor{green!30!black}{0}}  \\
OBW     &                    & $3.117 \times 10^{7}$  & --    & --    & \textbf{\textcolor{red!40!black}{100\%}} \\
\hline
WN      & \multirow{7}{*}{2} & \textbf{\textcolor{green!30!black}{4.277}} & \textbf{\textcolor{green!30!black}{0.091}} & \textbf{\textcolor{green!30!black}{0.066}} & \textbf{\textcolor{green!30!black}{0}}  \\
IB      &                    & $1.544 \times 10^{8}$  & \textbf{\textcolor{green!30!black}{0.091}} & \textbf{\textcolor{green!30!black}{0.066}} & \textbf{\textcolor{green!30!black}{0}}  \\
OB      &                    & $5.431 \times 10^{5}$  & \textbf{\textcolor{green!30!black}{0.091}} & \textbf{\textcolor{green!30!black}{0.066}} & \textbf{\textcolor{green!30!black}{0}}  \\
IBN     &                    & \textbf{\textcolor{red!40!black}{$4.810 \times 10^{11}$}} & \textbf{\textcolor{red!40!black}{0.110}} & \textbf{\textcolor{red!40!black}{0.074}} & 26\% \\
OBN     &                    & $2.032 \times 10^{8}$  & --    & --    & \textbf{\textcolor{red!40!black}{100\%}} \\
IBW     &                    & $1.480 \times 10^{6}$  & \textbf{\textcolor{green!30!black}{0.091}} & \textbf{\textcolor{green!30!black}{0.066}} & \textbf{\textcolor{green!30!black}{0}}  \\
OBW     &                    & $2.986 \times 10^{6}$  & \textbf{\textcolor{green!30!black}{0.091}} & \textbf{\textcolor{green!30!black}{0.066}} & \textbf{\textcolor{green!30!black}{0}}  \\
\hline
WN      & \multirow{7}{*}{3} & \textbf{\textcolor{green!30!black}{7.133}} & \textbf{\textcolor{green!30!black}{0.088}} & \textbf{\textcolor{green!30!black}{0.054}} & \textbf{\textcolor{green!30!black}{0}}  \\
IB      &                    & $2.897 \times 10^{8}$  & 0.089 & \textbf{\textcolor{green!30!black}{0.054}} & \textbf{\textcolor{green!30!black}{0}}  \\
OB      &                    & $2.179 \times 10^{5}$  & \textbf{\textcolor{green!30!black}{0.088}} & \textbf{\textcolor{green!30!black}{0.054}} & \textbf{\textcolor{green!30!black}{0}}  \\
IBN     &                    & \textbf{\textcolor{red!40!black}{$8.447 \times 10^{11}$}} & \textbf{\textcolor{red!40!black}{0.111}} & \textbf{\textcolor{red!40!black}{0.056}} & 20\% \\
OBN     &                    & $8.698 \times 10^{7}$  & --    & --    & \textbf{\textcolor{red!40!black}{100\%}} \\
IBW     &                    & $2.859 \times 10^{6}$  & \textbf{\textcolor{green!30!black}{0.088}} & \textbf{\textcolor{green!30!black}{0.054}} & \textbf{\textcolor{green!30!black}{0}}  \\
OBW     &                    & $1.251 \times 10^{6}$  & \textbf{\textcolor{green!30!black}{0.088}} & \textbf{\textcolor{green!30!black}{0.054}} & \textbf{\textcolor{green!30!black}{0}}  \\
\hline
WN      & \multirow{7}{*}{4} & \textbf{\textcolor{green!30!black}{20.927}} & \textbf{\textcolor{green!30!black}{0.245}} & \textbf{\textcolor{green!30!black}{0.208}} & \textbf{\textcolor{green!30!black}{0}}  \\
IB      &                    & $9.827 \times 10^{8}$  & 0.253 & 0.212 & 6\%  \\
OB      &                    & $3.976 \times 10^{5}$  & \textbf{\textcolor{green!30!black}{0.245}} & \textbf{\textcolor{green!30!black}{0.208}} & \textbf{\textcolor{green!30!black}{0}}  \\
IBN     &                    & \textbf{\textcolor{red!40!black}{$6.774 \times 10^{12}$}} & \textbf{\textcolor{red!40!black}{0.396}} & \textbf{\textcolor{red!40!black}{0.252}} & 54\% \\
OBN     &                    & $2.636 \times 10^{8}$  & --    & --    & \textbf{\textcolor{red!40!black}{100\%}} \\
IBW     &                    & $7.942 \times 10^{6}$  & \textbf{\textcolor{green!30!black}{0.245}} & \textbf{\textcolor{green!30!black}{0.208}} & \textbf{\textcolor{green!30!black}{0}}  \\
OBW     &                    & $2.692 \times 10^{6}$  & \textbf{\textcolor{green!30!black}{0.245}} & \textbf{\textcolor{green!30!black}{0.208}} & \textbf{\textcolor{green!30!black}{0}}  \\
\hline
WN      & \multirow{7}{*}{5} & \textbf{\textcolor{green!30!black}{38.616}} & \textbf{\textcolor{green!30!black}{0.140}} & \textbf{\textcolor{green!30!black}{0.085}} & \textbf{\textcolor{green!30!black}{0}}  \\
IB      &                    & $4.074 \times 10^{9}$  & \textbf{\textcolor{red!40!black}{0.154}} & \textbf{\textcolor{red!40!black}{0.090}} & 24\% \\
OB      &                    & $3.103 \times 10^{5}$  & \textbf{\textcolor{green!30!black}{0.140}} & \textbf{\textcolor{green!30!black}{0.085}} & \textbf{\textcolor{green!30!black}{0}}  \\
IBN     &                    & \textbf{\textcolor{red!40!black}{$3.249 \times 10^{13}$}} & --    & --    & \textbf{\textcolor{red!40!black}{100\%}} \\
OBN     &                    & $1.451 \times 10^{8}$  & --    & --    & \textbf{\textcolor{red!40!black}{100\%}} \\
IBW     &                    & $1.782 \times 10^{7}$  & \textbf{\textcolor{green!30!black}{0.140}} & \textbf{\textcolor{green!30!black}{0.085}} & \textbf{\textcolor{green!30!black}{0}}  \\
OBW     &                    & $1.820 \times 10^{6}$  & \textbf{\textcolor{green!30!black}{0.140}} & \textbf{\textcolor{green!30!black}{0.085}} & \textbf{\textcolor{green!30!black}{0}}  \\
\hline
\end{tabular}
}
\end{table}
\subsection{Comparison with Subspace Predictive Control (SPC)}
\begin{table}[!tb]
\centering
\caption{(\emph{SPC}) Mean and std of RMSEs~\eqref{eq:RMSE} for each system. The results are invariant to the excitation.}
\label{tab:SPC_res}
{
\begin{tabular}{c c c c c c }
\hline
Sys. & 1 & 2 & 3 & 4 & 5\\
\hline
$\mu_{\text{RMSE}}$ & 0.077 & 0.091 & 0.088 & 0.245 &  0.140\\
\hline 
$\sigma_{\text{RMSE}}$ & 0.021 & 0.066 & 0.054 & 0.208 & 0.085\\
\hline
\end{tabular}
}\vspace{-.3cm}
\end{table}
%
%
%\begin{figure}[!tb]
%\centering
%\begin{tabular}{cc}
%    \hspace*{-.3cm} \includegraphics[scale=0.44]{Img/u_S5_Mean_CI95_SPC.pdf} &
%\includegraphics[scale=0.44]{Img/y_S5_Mean_CI95_SPC.pdf}
%\end{tabular}
%    \caption{(\emph{SPC}) Tracking performance: mean values (colored line) and standard deviation (shaded area) achieved in closed-loop over $50$ simulations with different initial conditions. Outputs are compared against the reference to be tracked (black).}
%    \label{fig:mean_CI_traj_SPC}
%\end{figure}
We now compare the performance of DeePC with that of a Subspace Predictive Controller (SPC)~\cite{favoreel1999spc} for the different data collections already considered. As shown in Table \ref{tab:SPC_res}, SPC achieves the same tracking performance independently of the probing signal for all the systems. This result might be linked to the fact that the SPC predictor is the result of a least-squares estimation procedure (see \cite[Section 3.2]{favoreel1999spc}), with the data being lightly preprocessed before they are employed for control design. Moreover, by comparing these results with those in Table~\ref{tab:DeePC_res}, we obtain that DeePC and SPC attain the same performance whenever DeePC is feasible, in line with existing theoretical results (see~\cite{fiedler2021relationship}).\\  
These results thus suggest that data collection must fulfill additional requirements, rooted in the spectral bounds in \eqref{eq:leakagebounds}. The bounds indeed %\textcolor{brown}{These requirements are directly r \eqref{eq:leakagebound1}-\eqref{eq:leakagebound2}, which 
govern direct in-band excitation while dictating whether out-of-band system dynamics are sufficiently excited via leakage.
\subsection{Varying dataset length and band position}\label{Sec:Varying_data_and_pos}
We now analyze the sensitivity of the DeePC scheme to dataset size and the location of the excitation bandwidth by focusing on System 1 and System 4 only. In particular, we consider three different scenarios for $T_d$, starting from the minimum required to satisfy the PE condition in Definition~\ref{def:persistency_excitation_rank}, namely 
\begin{equation}\label{eq:min_data}
    T_d \geq (m+1)(\rho+L+n),
\end{equation}
which is $28$ for System 1 and $32$ for System 4. We then increase the dimension of the dataset used for DeePC to four ($T_d = 112$ and $T_d = 128$ for System 1 and 4, respectively) and ten ($T_d = 280$ for System 1 and $T_d = 320$ for System 4) times this lower bound. Meanwhile, we also consider bandwidths covering $25\%$ and $30\%$ of the frequency spectrum $(0,\pi]$ for System 1 and System 4, respectively, and evaluate three configurations. One starting at the first DFT grid point (B1) and three shifted by $0.2\pi$ (B2), $0.5\pi$ (B3), and $0.7\pi$ (B4), respectively. Note that we are not considering shifting the bandwidth by $0.7\pi$ for System 4 with $T_d =32$ since it exceeds the system's Nyquist frequency.

Tables~\ref{tab:T_d_band_S1} and~\ref{tab:T_d_band_S4} summarize the closed-loop results with a prediction horizon $L=4$. Crucially, meeting the minimal PE data requirement yields condition numbers at least $13$ orders of magnitude above the WN baseline across all bands for both systems, resulting in zero feasible simulations. Moreover, while increasing $T_d$ generally mitigates ill-conditioning, the rate of improvement slows down significantly as $T_d$ increases, even yielding some minor degradations. These results underscore the importance of bandwidth placement. Notably, increasing the dataset size improves conditioning and performance for B1 and B3 in System 1, and for B1 in System 4, while remaining ineffective for B4 across both systems. This difference seems to suggest that B1 and B3 in System 1, and B1 in System 4, are spectrally situated to cause weak out-of-band excitation, in turn potentially facilitated by spectral leakage as data accumulate, that partially excites the underlying system dynamics. Focusing on B1 and B4 for both Systems, we have that although B1 and B4 lie at comparable spectral distances relative to the opposite end of the frequency spectrum. Nonetheless, the system gain is much higher near B1 than near B4. Consequently, in line with \eqref{eq:leakagebound2}, the higher gain around B1 allows for accumulated spectral leakage to excite the dynamics as $T_d$ expands. Meanwhile, the lower gain near B4 suppresses out-of-band energy transfer, making it ineffective to increase the dataset dimension. In contrast, B2 for Systems 1 and 4 and B3 for System 4 seem to be spectrally positioned to generate sufficient spectral leakage even at a sample count of $T_d=112$, which already leads to the baseline performance. 
\begin{table}[!tb]
\centering
\caption{Performance indicators for System 1 across a $25\%$ bandwidth in $(0,\pi]$ with varying location and data length $T_d$.}
\label{tab:T_d_band_S1}
{
\begin{tabular}{c | c | c | c c c}
\hline
Case & $T_d$ & $\kappa$ & $\mu_{\text{RMSE}}$ & $\sigma_{\text{RMSE}}$ & $\pazocal{I}$ \\
\hline
WN                  & 28 & 13.002                 & 0.171 & 0.007 & 0  \\
\hline
\multirow{3}{*}{B1} & 28  & $5.834 \times 10^{14}$ & --    & --    & 100\% \\
                    & 112 & $4.795 \times 10^{7}$  & 0.200 & 0.016 & 94\% \\
                    & 280 & $4.438 \times 10^{7}$  & 0.186 & 0.025 & 46\% \\
\hline
\multirow{3}{*}{B2} & 28  & $1.254 \times 10^{15}$ & --    & --    & 100\% \\
                    & 112 & $1.344 \times 10^{5}$  & 0.171 & 0.007 & 0  \\
                    & 280 & $1.581 \times 10^{5}$  & 0.171 & 0.007 & 0  \\
\hline
\multirow{3}{*}{B3} & 28  & $6.032 \times 10^{14}$ & --    & --    & 100\% \\
                    & 112 & $9.484 \times 10^{7}$  & 0.171 & 0.007 & 64\% \\
                    & 280 & $2.505 \times 10^{7}$  & 0.171 & 0.007 & 0  \\
\hline
\multirow{3}{*}{B4} & 28  & $3.470 \times 10^{15}$ & --    & --    & 100\% \\
                    & 112 & $4.019 \times 10^{11}$ & --    & --    & 100\% \\
                    & 280 & $1.406 \times 10^{10}$ & --    & --    & 100\% \\
\hline
\end{tabular}
}
\end{table}
\begin{table}[!tb]
\centering
\caption{Performance indicators for System 4 across a $30\%$ bandwidth in $(0,\pi]$ with varying location and data length $T_d$.}
\label{tab:T_d_band_S4}
{
\begin{tabular}{c | c | c | c c c}
\hline
Case & $T_d$ & $\kappa$ & $\mu_{\text{RMSE}}$ & $\sigma_{\text{RMSE}}$ & $\pazocal{I}$ \\
\hline
WN                  & 32  & 14.220                 & 0.206 & 0.189 & 0  \\
\hline
\multirow{3}{*}{B1} & 32  & $2.026 \times 10^{15}$ & --    & --    & 100\% \\
                    & 128 & $2.881 \times 10^{8}$  & --    & --    & 100\% \\
                    & 320 & $2.304 \times 10^{8}$  & 0.249 & 0.211 & 4\%  \\
\hline
\multirow{3}{*}{B2} & 32  & $2.388 \times 10^{14}$ & --    & --    & 100\% \\
                    & 128 & $8.890 \times 10^{4}$  & 0.244 & 0.208 & 0  \\
                    & 320 & $5.934 \times 10^{4}$  & 0.244 & 0.208 & 0  \\
\hline
\multirow{3}{*}{B3} & 32  & $1.126 \times 10^{14}$ & --    & --    & 100\% \\
                    & 128 & $1.729 \times 10^{6}$  & 0.244 & 0.208 & 0  \\
                    & 320 & $3.099 \times 10^{6}$  & 0.244 & 0.208 & 0  \\
\hline
\multirow{2}{*}{B4} & 128 & $2.384 \times 10^{9}$  & --    & --    & 100\% \\
                    & 320 & $1.270 \times 10^{10}$ & --    & --    & 100\% \\
\hline
\end{tabular}
}
\end{table}
\subsection{Varying dataset length and past horizon}
At last, we evaluate how changing the past horizon length $\rho \in \{4,6,8\}$ affects performance for a fixed bandwidth. We specifically focus on System 1 under band B1, varying the dataset size from the minimal PE length (see \eqref{eq:min_data}) up to four times this threshold. Table~\ref{tab:rho_data} highlights that for larger values of $\rho$, data availability becomes crucial to ensure optimization feasibility. This result is expected, since increasing $\rho$ makes \eqref{eq:rank_condition} more demanding. Meanwhile, as before, increasing $T_d$ allows spectral leakage to accumulate additional out-of-band energy, with an improvement that remains intrinsically bounded by the bandwidth spectral distance to the system's poles and zeros. 
\begin{table}[!tb]
\centering
\caption{Performance indicators for System 1 with excitation bandwidth B1 varying the past horizon length $\rho$ and data amount $T_d$.}
\label{tab:rho_data}
{
\renewcommand{\arraystretch}{1.3}
\begin{tabular}{c | c | c | c c c}
\hline
$\rho$ & $T_d$ & $\kappa$ & $\mu_{\text{RMSE}}$ & $\sigma_{\text{RMSE}}$ & $\pazocal{I}$ \\
\hline
\multirow{5}{*}{4} & 20  & $2.916 \times 10^{14}$ & --    & --    & 100\% \\
                   & 40  & $1.575 \times 10^{5}$  & 0.171 & 0.006 & 0  \\
                   & 60  & $4.096 \times 10^{4}$  & 0.171 & 0.006 & 0  \\
                   & 80  & $6.621 \times 10^{4}$  & 0.171 & 0.006 & 0  \\
                   & 100 & $5.017 \times 10^{4}$  & 0.171 & 0.006 & 0  \\
\hline
\multirow{5}{*}{6} & 24  & $3.623 \times 10^{14}$ & --    & --    & 100\% \\
                   & 48  & $6.260 \times 10^{6}$  & 0.171 & 0.007 & 4\%  \\
                   & 72  & $2.103 \times 10^{6}$  & 0.171 & 0.006 & 0  \\
                   & 96  & $2.289 \times 10^{6}$  & 0.171 & 0.006 & 0  \\
                   & 120 & $1.697 \times 10^{6}$  & 0.171 & 0.006 & 0  \\
\hline
\multirow{5}{*}{8} & 28  & $6.929 \times 10^{14}$ & --    & --    & 100\% \\
                   & 56  & $3.977 \times 10^{8}$  & 0.253 & 0.025 & 76\% \\
                   & 84  & $2.756 \times 10^{7}$  & 0.180 & 0.013 & 74\% \\
                   & 112 & $5.582 \times 10^{7}$  & 0.230 & 0.024 & 74\% \\
                   & 140 & $3.189 \times 10^{7}$  & 0.184 & 0.024 & 56\% \\
\hline
\end{tabular}
}
\end{table}
\section{Conclusions}\label{sec:conclusions}
By analyzing the predictor at the core of data-enabled predictive control from a frequency perspective, this work has provided insights into the impact of limited excitation on the features of the predicted inputs and outputs, as well as on closed-loop performance. Our results highlight the importance of accounting for bandwidth positioning and width when collecting data.\\
Future research directions indeed include extending these insights to a noisy setting, and formally investigating the link between data matrix conditioning and an excitation band.

\bibliography{Main.bib}

@article{willems2005note,
  title={A note on persistency of excitation},
  author={Willems, Jan C and Rapisarda, Paolo and Markovsky, Ivan and De Moor, Bart LM},
  journal={Systems \& Control Letters},
  volume={54},
  number={4},
  pages={325--329},
  year={2005},
  publisher={Elsevier}
}

@book{pintelon2012system,
  title={System identification: a frequency domain approach},
  author={Pintelon, Rik and Schoukens, Johan},
  year={2012},
  publisher={John Wiley \& Sons}
}

@article{vakilzadeh2015experiment,
  title={Experiment design for improved frequency domain subspace system identification of continuous-time systems},
  author={Vakilzadeh, Majid Khorsand and Yaghoubi, Vahid and McKelvey, Tomas and Abrahamsson, Thomas and Ljung, Lennart},
  journal={IFAC-PapersOnLine},
  volume={48},
  number={28},
  pages={886--891},
  year={2015},
  publisher={Elsevier}
}

@inproceedings{coulson2019data,
  title={Data-enabled predictive control: In the shallows of the DeePC},
  author={Coulson, Jeremy and Lygeros, John and D{\"o}rfler, Florian},
  booktitle={2019 18th European Control Conference (ECC)},
  pages={307--312},
  year={2019},
  organization={IEEE}
}

@article{berberich2020data,
  title={Data-driven model predictive control with stability and robustness guarantees},
  author={Berberich, Julian and K{\"o}hler, Johannes and M{\"u}ller, Matthias A and Allg{\"o}wer, Frank},
  journal={IEEE Transactions on Automatic Control},
  volume={66},
  number={4},
  pages={1702--1717},
  year={2020},
  publisher={IEEE}
}

@article{breschi2023data,
  title={Data-driven predictive control in a stochastic setting: a unified framework},
  author={Breschi, Valentina and Chiuso, Alessandro and Formentin, Simone},
  journal={Automatica},
  volume={152},
  pages={110961},
  year={2023},
  publisher={Elsevier}
}

@article{dorfler2022bridging,
  title={Bridging direct and indirect data-driven control formulations via regularizations and relaxations},
  author={D{\"o}rfler, Florian and Coulson, Jeremy and Markovsky, Ivan},
  journal={IEEE Transactions on Automatic Control},
  volume={68},
  number={2},
  pages={883--897},
  year={2022},
  publisher={IEEE}
}

@article{yin2021maximum,
  title={Maximum likelihood estimation in data-driven modeling and control},
  author={Yin, Mingzhou and Iannelli, Andrea and Smith, Roy S},
  journal={IEEE Transactions on Automatic Control},
  volume={68},
  number={1},
  pages={317--328},
  year={2021},
  publisher={IEEE}
}

@book{van2012subspace,
  title={Subspace identification for linear systems: Theory-Implementation-Applications},
  author={Van Overschee, Peter and De Moor, BL0888},
  year={2012},
  publisher={Springer Science \& Business Media}
}

@book{ljung1987theory,
  title={System identification: Theory for the user},
  author={Ljung, Lennart},
  year={1999},
  publisher={Prentice-Hall, Inc.}
}

@article{narasimhan2003multi,
  title={Multi-objective input signal design for plant-friendly identification},
  author={Narasimhan, Sridharakumar and Srinivasan, Ranganathan and Rengaswamy, Raghunathan},
  journal={IFAC Proceedings Volumes},
  volume={36},
  number={16},
  pages={897--902},
  year={2003},
  publisher={Elsevier}
}

@article{rivera2003plant,
  title={"{Plant-Friendly}" system identification: a challenge for the process industries},
  author={Rivera, Daniel E and Lee, Hyunjin and Braun, Martin W and Mittelmann, Hans D},
  journal={IFAC Proceedings Volumes},
  volume={36},
  number={16},
  pages={891--896},
  year={2003},
  publisher={Elsevier}
}

@article{van2021beyond,
  title={Beyond persistent excitation: Online experiment design for data-driven modeling and control},
  author={van Waarde, Henk J},
  journal={IEEE Control Systems Letters},
  volume={6},
  pages={319--324},
  year={2021},
  publisher={IEEE}
}

@article{alsalti2023design,
  title={On the design of persistently exciting inputs for data-driven control of linear and nonlinear systems},
  author={Alsalti, Mohammad and Lopez, Victor G and M{\"u}ller, Matthias A},
  journal={IEEE Control Systems Letters},
  volume={7},
  pages={2629--2634},
  year={2023},
  publisher={IEEE}
}

@article{gevers2006input,
  title={Input design: From open-loop to control-oriented design},
  author={Gevers, Michel and Bombois, Xavier},
  journal={IFAC Proceedings Volumes},
  volume={39},
  number={1},
  pages={1329--1334},
  year={2006},
  publisher={Elsevier}
}

@article{hjalmarsson2005experiment,
  title={From experiment design to closed-loop control},
  author={Hjalmarsson, H{\aa}kan},
  journal={Automatica},
  volume={41},
  number={3},
  pages={393--438},
  year={2005},
  publisher={Elsevier}
}

@inproceedings{iannelli2021design,
  title={Design of input for data-driven simulation with Hankel and Page matrices},
  author={Iannelli, Andrea and Yin, Mingzhou and Smith, Roy S},
  booktitle={2021 60th IEEE Conference on Decision and Control (CDC)},
  pages={139--145},
  year={2021},
  organization={IEEE}
}

@article{favoreel1999spc,
  title={{SPC}: Subspace predictive control},
  author={Favoreel, Wouter and De Moor, Bart and Gevers, Michel},
  journal={IFAC Proceedings Volumes},
  volume={32},
  number={2},
  pages={4004--4009},
  year={1999},
  publisher={Elsevier}
}

@inproceedings{fiedler2021relationship,
  title={On the relationship between data-enabled predictive control and subspace predictive control},
  author={Fiedler, Felix and Lucia, Sergio},
  booktitle={2021 European Control Conference (ECC)},
  pages={222--229},
  year={2021},
  organization={IEEE}
}

@article{coulson2022quantitative,
  title={A quantitative notion of persistency of excitation and the robust fundamental lemma},
  author={Coulson, Jeremy and Van Waarde, Henk J and Lygeros, John and D{\"o}rfler, Florian},
  journal={IEEE Control Systems Letters},
  volume={7},
  pages={1243--1248},
  year={2022},
  publisher={IEEE}
}

@article{zhou2025effective,
  title={The Effective Singular Value Ratio: A Novel Criterion for Adaptive Data Selection in {DeePC}},
  author={Zhou, Yuanqiang and Wang, Mingrui and He, Yuhang and Xu, Jia and Li, Dewei and Chen, Hong},
  journal={IEEE Control Systems Letters},
  volume={9},
  pages={2675--2680},
  year={2025},
  publisher={IEEE}
}

@misc{grant2008cvx,
  title={{CVX}: Matlab software for disciplined convex programming},
  author={Grant, Michael and Boyd, Stephen and Ye, Yinyu},
  year={2008},
  publisher={Stanford University, Stanford, CA, USA}
}

@misc{gurobi,
  author = {{Gurobi Optimization, LLC}},
  title = {{Gurobi Optimizer Reference Manual}},
  year = 2026,
}
\end{document}